\documentclass{IEEEtran}
\usepackage{cite}
\usepackage{amsmath,amssymb,amsfonts,amsthm,mathrsfs,comment,multirow}
\usepackage[inline]{enumitem}
\usepackage[ruled,vlined,algo2e]{algorithm2e}
\newtheorem{theorem}{Theorem}
\newtheorem{lemma}{Lemma}
\newtheorem{remark}{Remark}
\newtheorem{definition}{Definition}
\newtheorem{assumption}{Assumption}

\newtheorem{fact}{Fact}
\usepackage{graphicx}
\usepackage{textcomp}
\def\BibTeX{{\rm B\kern-.05em{\sc i\kern-.025em b}\kern-.08em
    T\kern-.1667em\lower.7ex\hbox{E}\kern-.125emX}}
\begin{document}
\title{Adaptive Optimal Control of Linear Periodic Systems: An Off-Policy Value Iteration Approach}
\author{Bo Pang, \IEEEmembership{Student Member, IEEE}, and Zhong-Ping Jiang, \IEEEmembership{Fellow, IEEE}
\thanks{This work was partially supported by the National Science Foundation under Grants ECCS-1501044 and EPCN-1903781.}
\thanks{The authors are with the Control and Networks Lab, Department of Electrical and Computer Engineering, Tandon School of Engineering, New York University, 370 Jay Street, Brooklyn, NY 11201, USA (e-mail: bo.pang@nyu.edu; zjiang@nyu.edu).}
}
\renewcommand{\arraystretch}{1.2}

\maketitle

\begin{abstract}
This paper studies the infinite-horizon adaptive optimal control of continuous-time linear periodic (CTLP) systems. A novel value iteration (VI) based off-policy ADP algorithm is proposed for a general class of CTLP systems, so that approximate optimal solutions can be obtained directly from the collected data, without the exact knowledge of system dynamics. Under mild conditions, the proofs on uniform convergence of the proposed algorithm to the optimal solutions are given for both the model-based and model-free cases. The VI-based ADP algorithm is able to find suboptimal controllers without assuming the knowledge of an initial stabilizing controller. Application to the optimal control of a triple inverted pendulum subjected to a periodically varying load demonstrates the feasibility and effectiveness of the proposed method.
\end{abstract}


\section{Introduction}
Recently, reinforcement learning (RL) has invoked a lot of research interests from both researchers in academia and practitioners in industry, due to its successful applications to the design of intelligent computer GO player and many other intelligent agents learning tasks \cite{sutton2018reinforcement}. In RL, agents (or controllers) find optimal decisions (controls) from scratch through its interactions with an unknown environment. In spite of its popularity, most previous RL algorithms have their own limitations. Firstly, the underlying environments are described by Markov decision processes \cite{sutton2018reinforcement}, where time is discrete and the state and input spaces are finite or countable. Secondly, stability and safety properties associated with the use of the obtained optimal policy are not considered and guaranteed. However, many physical systems are more naturally described by differential equations, where time is continuous and the state and input spaces are infinite and continuous. Stability and safety are also indispensable considerations in real-world applications, e.g., autonomous vehicles. To this end, over the past decade, another stream of RL algorithms has emerged, to solve the optimal control problems described by differential equations, without the exact knowledge of the system dynamics, and with stability guarantees; see, e.g., \cite{jiang2017robust}, \cite{6315769} and numerous references therein. This class of RL algorithms are often coined adaptive dynamic programming (ADP), to be distinguished from those with Markov decision processes. The interested reader can consult the books \cite{jiang2017robust}, \cite{Frank-book1} for several practical applications of ADP.

While significant progresses have been made in the development of ADP, most of the existing results are devoted exclusively to time-invariant systems. When problems arise from applications involving time-varying control systems, those ADP algorithms previously developed for time-invariant systems are not directly applicable. Recently, in \cite{FONG201849} and \cite{BoCTT}, the finite-horizon optimal control problem was studied for time-varying systems by ADP. However, the corresponding infinite-horizon optimal control problem for time-varying systems described by differential equations has received scanty attention. There are several technical obstacles for this generalization. First, the stability analysis and control synthesis of time-varying systems are much more challenging than the case of time-invariant systems. Second, predicting the future evolution of the system trajectories becomes an intractable task for general time-varying systems, using only the historical data collected over a finite period of time, which is a key to the development of ADP algorithms. With these observations in mind, how to develop ADP algorithms to address the infinite-horizon optimal control problem of uncertain time-varying systems with guaranteed stability remains an open problem. In this paper, we take a step forward to study this longstanding unresolved issue. To this end, we will examine the infinite-horizon adaptive optimal control of continuous-time linear periodic (CTLP) systems. The analysis and control of linear period systems have played an important role in various applications. By exploiting the periodic time-varying nature, vibration is significantly actively suppressed in wind turbine system \cite{houtzager2013wind} and rotor-blade system \cite{camino2019periodic}; through the design of periodic model predictive control strategies for periodic systems, better economic performance is achieved in building climate control \cite{gondhalekar2013least}, drinking water network \cite{limon2014single} and nonisolated microgrid \cite{pereira2015periodic}; periodic feedback controller is reported to outperform a standard time-invariant feedback controller in online advertising \cite{karlsson2018control}; to name a few. Orbital stabilization of time-invariant nonlinear systems can also be analyzed and designed using linear periodic systems, since linearization of nonlinear systems along a periodic orbit yields linear periodic systems \cite[Section 5.1]{farkas2013periodic}. It should be emphasized that even for the class of CTLP systems, the design of ADP algorithm is a non-trivial task, as a result of the nonlinear dependence of system parameters on the time. 

Inspired by the time-invariant results in \cite{7798777}, a novel value iteration (VI) based ADP algorithm is proposed for a class of CTLP systems in this paper, to find approximate optimal controllers without the exact knowledge of system dynamics and an initial stabilizing controller. The VI-based ADP is based on the asymptotic property of finite-horizon solution of the periodic Riccati equation (PRE). It is claimed in \cite{doi:10.1137/0313077} that the solution of the PRE starting from a positive semidefinite initial matrix converges to the stabilizing solution of the same PRE, under certain conditions. However, it is pointed out by the authors of \cite{doi:10.1137/0316003} and \cite{Bittanti1991} that the proof of the claim in \cite{doi:10.1137/0313077} is based on some wrong preliminary results (see Remark \ref{hewer_wrong}). In the present paper, we firstly give a new proof of the claim, and then present a VI-based ADP algorithm, using the Fourier basis approximation. It turns out that the VI-based ADP algorithm amounts to numerically solving the final value problem of a nonlinear differential equation, which only involves collected data and is independent of the exact system dynamics. In Section \ref{section_VI}, the uniform convergence of the VI-based ADP algorithm to the optimal solution of the corresponding optimal control problem is rigorously proved. In Section \ref{section_Sim}, the proposed VI-based ADP algorithm is applied to the adaptive optimal control of a triple inverted pendulum subjected to a periodically varying load, which demonstrates the effectiveness of the resulting algorithm. Section V closes the paper with some concluding remarks.

It is worth noting that there is a rich literature on optimal control (see, e.g., \cite{houtzager2013wind,camino2019periodic, gondhalekar2013least, limon2014single,pereira2015periodic}) and on adaptive control (see, e.g., \cite{4782000,1284721,narendra2019adaptive}) for linear periodic systems. However, they have been studied as two separate problems. That is, the optimal control solutions presented in \cite{houtzager2013wind,camino2019periodic, gondhalekar2013least, limon2014single,pereira2015periodic} require the precise knowledge of the system dynamics, while the adaptive control results presented in \cite{narendra2019adaptive,1284721,4782000} do not guarantee optimization of any prescribed cost function. Different from both groups of research, our proposed method finds the suboptimal solutions directly from the input/state data, without the exact knowledge of the system dynamics.

\textbf{Notations}: $\mathbb{R}$ ($\mathbb{R}_+$) is the set of (nonnegative) real numbers. $\mathbb{Z}_+$ is the set of nonnegative integers. $\mathbb{S}^{n}$ denotes the vector space of all $n$-by-$n$ real symmetric matrices. $\otimes$ is the Kronecker product operator. $\vert \cdot \vert$ and $\Vert\cdot\Vert$ represent the Euclidean norm for vectors and the Frobenius norm for matrices, respectively. $[v]_{j}$ denotes the $j$th element of vector $v\in\mathbb{R}^{n}$. $[X]_{i,j}$ denotes the element in $i$th row and $j$th column of matrix $X\in\mathbb{R}^{m\times n}$. $\lfloor v\rfloor$ represents the largest integer no larger than $v\in \mathbb{R}$. $X^\dagger$ denotes the Moore-Penrose inverse of matrix $X$. $\sigma_{\min}(X)$ is the minimal singular value of matrix $X$.

\section{Problem Formulation and Preliminaries}
Consider continuous-time linear periodic systems
\begin{equation}\label{periodSYS}
\dot{x}(t) = A(t)x(t)+B(t)u(t),
\end{equation}
where $x(t)\in\mathbb{R}^n$ is the system state, $u(t)\in\mathbb{R}^m$ is the control input, $A(\cdot):\mathbb{R}\rightarrow \mathbb{R}^{n\times n}$, $B(\cdot):\mathbb{R}\rightarrow \mathbb{R}^{n\times m}$ are continuous and $T$-periodic matrix-valued functions, i.e.,
$$A(t+T)=A(t),\quad B(t+T)=B(t),\quad T\in\mathbb{R}_+,\quad \forall t\in \mathbb{R}.$$
Let $\Phi(t,t_0)$, $t>t_0$, $t_0\in\mathbb{R}$ denote the state transition matrix of the unforced system of (\ref{periodSYS}), with $u=0$. In the setting of linear periodic system, the matrix $\Phi(t_0+T,t_0)$ is known as the monodromy matrix. Its eigenvalues (also called characteristic multipliers) are independent of $t_0$. $A(\cdot)$ is asymptotically stable if and only if its characteristic multipliers are inside the open unit disk. See \cite{10.1007/BFb0043803} for the details.

The infinite-horizon periodic linear quadratic (PLQ) optimal control problem \cite[Section 6.5.1.1]{Bittanti1991} is to find a linear stabilizing control law
$u(t)=-K(t)x(t),$
where $K(\cdot):\mathbb{R}\rightarrow\mathbb{R}^{m\times n}$ is continuous and $T$-periodic, such that the following quadratic cost is minimized
\begin{equation}\label{cost_inf}
    J(t_0,\xi,u(\cdot)) =  \int_{t_0}^\infty{\left(\left\vert C(t)x(t)\right\vert^2+ u^T(t)R(t)u(t)\right)}dt,
\end{equation}
where $C(\cdot): \mathbb{R}\rightarrow\mathbb{R}^{r\times n}$ is continuous and $T$-periodic; $R(\cdot): \mathbb{R}\rightarrow\mathbb{R}^{m\times m}$ is continuous, $T$-periodic, positive definite and piecewise continuously differentiable; $x(t)$ is the solution of (\ref{periodSYS}) with initial state $x(t_0)=\xi$, $\xi\in \mathbb{R}^n$. Associated with the PLQ control problem is the PRE
\begin{equation}\label{PRE}
\begin{split}
    -\dot{P}(t) &= A^T(t)P(t)+P(t)A(t) \\
    &-P(t)B(t)R^{-1}(t)B^T(t)P(t)+C^T(t)C(t).
\end{split}
\end{equation}
Generally, the PRE (\ref{PRE}) may admit many different kinds of solutions, among which two particular kinds  are relevant to this paper.
\begin{definition}[{\cite{Bittanti1991,doi:10.1080/00207179208934305}}]\label{PRE_solutions_definition}
Consider the real symmetric, periodic and positive semidefinite (SPPS) solutions satisfying PRE (\ref{PRE}) over time interval $(-\infty,\infty)$.
\begin{enumerate}
\item $P_S(\cdot)$ is called strong solution, if the characteristic multipliers of $D_S(t)=A(t)-B(t)R^{-1}(t)B^T(t)P_S(t)$ belong to the closed unit disk.
\item $P_+(\cdot)$ is called stabilizing solution, if the characteristic multipliers of $D_+(t)=A(t)-B(t)R^{-1}(t)B^T(t)P_+(t)$ belong to the open unit disk.
\end{enumerate}
\end{definition}
\begin{assumption}\label{structure_assum}
$(A(\cdot),B(\cdot))$ is stabilizable and $(A(\cdot),C(\cdot))$ is detectable \cite[Theorem 4]{10.1007/BFb0043803}.
\end{assumption}
Under Assumption \ref{structure_assum}, the optimal solution to the infinite-horizon PLQ control problem exists and is unique \cite[Theorem 6.5 and 6.12]{Bittanti1991}.
\begin{lemma}\label{UniqueExist}
There exists a unique SPPS solution $P^*(\cdot)$ of the PRE, and the corresponding closed-loop system is stable, if and only if Assumption \ref{structure_assum} is satisfied. In addition, \begin{enumerate*}[label=\alph*)]
\item $P_S=P_+=P^*$. \item the cost (\ref{cost_inf}) is minimized by the optimal controller $u^* (t)= -K^*(t)x(t)$, with $K^*(t)=R^{-1}(t)B^T(t)P^*(t)$. \item the corresponding minimum cost is $J^*(t_0,\xi)=J(t_0,\xi,u^*(\cdot))=\xi^TP^*(t_0)\xi$. 
\end{enumerate*}  
\end{lemma}

In general, it is difficult to obtain an analytic expression for $P^*(\cdot)$, which is a nonlinear matrix-valued function of time $t$. In this paper, Fourier basis functions are adopted to approximate different periodic functions. For a continuous and $T$-periodic function $f(\cdot):\mathbb{R}\rightarrow\mathbb{R}$, partial sums of its Fourier series representation are
\begin{equation*}
f_N(x) = \frac{a_0}{2}+\sum_{i=1}^N\left(a_i\cos{(\omega ix)}+b_i\sin{(\omega ix)}\right),
\end{equation*}
where $\omega = 2\pi/T$, $N\in\mathbb{Z}_+$, $\{a_i\}_{i=0}^N$ and $\{b_i\}_{i=1}^N$ are Fourier coefficients. The following lemma gives the asymptotic property of using $f_N$ to  approximate $f$.
\begin{lemma}[{\cite[Theorem 1.5.1]{FourierBook}}]\label{FourierConverge}
If $f$ is $T$-periodic, continuous and piecewise continuously differentiable, then $f_N\rightarrow f$ uniformly, as $N \rightarrow \infty$.
\end{lemma}

When matrices $A(\cdot)$ and $B(\cdot)$ are unknown, the optimal solution $P^*(\cdot)$ can hardly be obtained directly due to the nonlinearity of the PRE. In next section, VI is exploited to find approximate optimal controllers directly from the input/state data collected along the controlled system trajectories. As it can be directly checked, we have
\begin{fact}\label{Isometric}
For $X\in\mathbb{R}^{n\times m}$, $Y\in\mathbb{S}^n$, $v\in\mathbb{R}^n$,
$$\vert\mathrm{vec}(X)\vert=\Vert X\Vert,\ \vert\mathrm{vecs}(Y)\vert=\Vert Y\Vert,\ v^TYv\equiv \tilde{v}^T\mathrm{vecs}(Y),$$ where \begin{align*}
    \mathrm{vec}(X) &= [X_1^T, X_2^T, \cdots, X_m^T]^T, \\
    \mathrm{vecs}(Y) &= [y_{11},\sqrt{2}y_{12},\cdots,\sqrt{2}y_{1m},y_{22},\sqrt{2}y_{23}, \\ 
    &\cdots,\sqrt{2}y_{m-1,m},y_{m,m}]^T\in \mathbb{R}^{\frac{1}{2}m(m+1)}, \\
    \tilde{v} &= [v_1^2,\sqrt{2}v_1v_2,\cdots,\sqrt{2}v_1v_n,v_2^2,\sqrt{2}v_2v_3, \\ 
    &\cdots,\sqrt{2}v_{n-1}v_n,v_n^2]^T\in \mathbb{R}^{\frac{1}{2}n(n+1)},
\end{align*}
$X_i$ is the $i$th column of $X$. In addition, there always exist operations $\mathrm{vec}^{-1}(\cdot)$ and $\mathrm{vecs}^{-1}(\cdot)$, such that $X = \mathrm{vec}^{-1}(\mathrm{vec}(X))$ and $Y = \mathrm{vecs}^{-1}{(\mathrm{vecs}(Y))}$, respectively.
\end{fact}

\section{Value Iteration based Adaptive Dynamic Programming for Continuous-time Linear Periodic Systems}\label{section_VI}

The value iteration method is based on the asymptotic property of the solution to the finite-horizon PLQ optimal control problem. For any $t<t_f$ and a measurable locally essentially bounded input $u:[t,t_f)\rightarrow \mathbb{R}^m$, define finite-horizon cost
\begin{equation*}
\begin{split}
   &V(t,t_f,\xi,u(\cdot),G) =   x(t_f)^TGx(t_f)\\&+
   \int_{t}^{t_f}{\left(\left\vert C(t)x(t)\right\vert^2+u(t)^TR(t)u(t)\right)}dt,
\end{split}
\end{equation*}
for all $\xi\in \mathbb{R}^n$ and $G\in\mathbb{S}^{n}$, $G\geq 0$, where $x(t)=\xi$.
Starting at $P(t_f)=G$, the corresponding solution of the PRE (\ref{PRE}) at time $t<t_f$, denoted by $P(t;t_f,G)$, satisfies
\begin{equation*}
\xi^TP(t;t_f,G)\xi = \min_u{V(t,t_f,\xi,u(\cdot),G)}.
\end{equation*}
Generally, $P(\cdot;t_f,G)$ is not necessarily periodic, and different from the SPPS solutions in Definition \ref{PRE_solutions_definition}, $P(\cdot;t_f,G)$ satisfies PRE (\ref{PRE}) over time interval $(-\infty,t_f]$ \cite[Section 6.1.4]{10.2307/j.ctvcm4g0s}. Next, it is shown that $P(t;t_f,G)$ with $G\geq 0$ will approach the SPPS solution $P^*(t)$ of the PRE (\ref{PRE}), as time $t\rightarrow-\infty$.
\begin{lemma}\label{VI_monotone}
For any $0\leq G_1\leq G_2$ and $t<t_f$,
$$P(t;t_f,G_1)\leq P(t;t_f,G_2).$$
\end{lemma}
\begin{proof}
For any fixed $\xi\in \mathbb{R}^n$, and any measurable locally essentially bounded $u$,
$$\xi^TP(t;t_f,G_1)\xi\leq V(t,t_f,\xi,u,G_1)\leq V(t,t_f,\xi,u,G_2).$$
Minimizing the above inequalities simultaneously over $u$, we obtain
$\xi^TP(t;t_f,G_1)\xi\leq \xi^T P(t;t_f,G_2)\xi$.
Since $\xi$ is arbitrary, the proof is completed.
\end{proof}
\begin{theorem}\label{VI}
Under Assumption \ref{structure_assum}, if $G=G^T\geq 0$, then
\begin{equation}\label{VImodel}
\lim_{t\rightarrow-\infty}\left(P(t;t_f,G)-P^*(t)\right)=0.
\end{equation}
\end{theorem}
\begin{proof}
See the Appendix.
\end{proof}
\begin{remark}\label{hewer_wrong}
In 1975, Hewer drew the same conclusion \cite[Theorem 4.11]{doi:10.1137/0313077} with Theorem \ref{VI}. However, as pointed out in \cite[Section 6.3.4]{Bittanti1991} and \cite{doi:10.1137/0316003}, the proof of \cite[Theorem 4.11]{doi:10.1137/0313077} was based on \cite[Theorem 3.7]{doi:10.1137/0313077}, which is shown wrong by a counterexample in \cite[Section 2]{doi:10.1137/0316003}. The proof of Theorem \ref{VI} in this paper is new and is included for the sake of completeness.
\end{remark}
Note that the solutions of PRE (\ref{PRE}) need not evolve according to the same time variable in system (\ref{periodSYS}). To emphasize this point, in the rest of this paper we use $s\in\mathbb{R}$ for the algorithmic time, which is the time used in the PRE, while $t\in\mathbb{R}$ is reserved for the system evolution time, i.e. the time used in system (\ref{periodSYS}). This separation of time variables is essential in the development of our proposed algorithm in the sequel (also see Remark \ref{remark_algorithmic_time}).

Theorem \ref{VI} means that near-optimal solutions of $P^*(\cdot)$ can be found by solving the PRE (\ref{PRE}) backward in time with boundary condition $G\geq 0$. Concretely, we can solve the following final value problem on interval $[0,s_f]$,
\begin{equation}\label{PRE_s}
\begin{split}
    -\dot{P}(s) &= A^T(s)P(s)+P(s)A(s)+C^T(s)C(s)\\
    &-P(s)B(s)R^{-1}(s)B^T(s)P(s), \quad P(s_f)=G,
\end{split}
\end{equation}
where $G=G^T\geq 0$, $P(s)$ is short for $P(s;s_f,G)$. By Theorem \ref{VI}, if $s_f>0$ is large, then $P(s)$ is close to $P^*(s)$ for $s$ near $0$. However, the \textit{a priori} knowledge of $A(\cdot)$ and $B(\cdot)$ is still required to solve (\ref{PRE_s}). Next, a VI-based off-policy ADP algorithm is proposed to solve (\ref{PRE_s}) directly from input/state data, without the exact knowledge of $A(\cdot)$ and $B(\cdot)$.

Define matrix-valued functions
\begin{equation}\label{VI_variables_def}
\begin{split}
    H(s,t) &= A^T(t)P(s)+P(s)A(t),\\
    K(s,t) &= R^{-1}(t)B^T(t)P(s).
\end{split}
\end{equation}
Then (\ref{PRE_s}) can be rewritten as,
\begin{equation}\label{PRE_HK}
\begin{split}
    -\dot{P}(s)&=H(s,s)+C^T(s)C(s)\\
    &-K^T(s,s)R(s)K(s,s).
\end{split}
\end{equation}
By Theorem \ref{VI}, as long as $s_f$ is large enough, $H(s,s)$ and $K(s,s)$ with $s$ near $0$ will be good approximation to
\begin{equation}\label{HK_star}
    \begin{split}
        H^*(s) &= A^T(s)P^*(s)+P^*(s)A(s), \\ 
        K^*(s) &= R^{-1}(s)B^T(s)P^*(s).
    \end{split}
\end{equation}
With Lemma \ref{UniqueExist}, note that $K^*(\cdot)$ is the optimal control gain.

For fixed $s\in [0,s_f]$, $H(s,t)$ and $K(s,t)$ are periodic with respect to time $t\in\mathbb{R}$. Thus we can express $\mathrm{vecs}(H(s,t))$ and $\mathrm{vec}(K(s,t))$ by their Fourier series
\begin{equation}\label{VI_Fourier_series}
\begin{split}
    \mathrm{vecs}(H(s,t)) &= W^{H}(s)F_N(t)+e^{H}_{N}(s,t),\\
    \mathrm{vec}(K(s,t)) &= W^{K}(s)F_N(t)+e^{K}_{N}(s,t),
\end{split}
\end{equation}
where $W^{H}(s)\in \mathbb{R}^{n_1\times(2N+1)}$, $n_1=n(n+1)/2$ and $W^{K}(s)\in\mathbb{R}^{n_2\times(2N+1)}$, $n_2=mn$ are Fourier coefficients at algorithmic time $s$,
\begin{equation*}
    \begin{split}
        F_N(t) &= \left[1, \cos{(\omega t)}, \sin{(\omega t)}, \cos{(2\omega t)}, \sin{(2\omega t)},\right.\\
        &\left.\cdots, \cos{(N\omega t)}, \sin{(N\omega t)}\right]^T,
    \end{split}
\end{equation*}
$e_N^{H}(s,t)\in\mathbb{R}^{n_1}$ and $e_N^{K}(s,t)\in\mathbb{R}^{n_2}$ are truncation errors.

Here is a preview of subsequent development of our algorithm. We aim at solving (\ref{PRE_s}) (i.e. (\ref{PRE_HK})) directly from input/state data. To this end, firstly, using (\ref{VI_variables_def}) and (\ref{VI_Fourier_series}), the PRE (\ref{PRE_s}) (i.e. (\ref{PRE_HK})) is transformed into an equivalent differential equation (\ref{VI_data_eqn3}), in terms of Fourier coefficients $W^H(s)$ and $W^K(s)$. Secondly, by ignoring the truncation errors $e_N^{H}(s,t)$ and $e_N^{K}(s,t)$ in (\ref{VI_data_eqn3}), a new differential equation (\ref{VI_adp_eqn1}) is derived in terms of $\hat{W}^H(s)$ and $\hat{W}^K(s)$, where only input/state data is involved. Thirdly, it is shown (in Lemma \ref{VI_lemma_param_converge}) that, $\hat{W}^H(s)$ and $\hat{W}^K(s)$ given by the new differential equation (\ref{VI_adp_eqn1}) are close to the original Fourier coefficients $W^H(s)$ and $W^K(s)$ on $[0,s_f]$ under suitable conditions, despite of the ignorance of the truncation errors. This implies that the problem of solving model-based PRE (\ref{PRE_s}) (i.e. (\ref{PRE_HK})) on $[0,s_f]$ can be approached by solving the data-based differential equation (\ref{VI_adp_eqn1}) on $[0,s_f]$. Thus then (\ref{VI_adp_eqn1}) is solved numerically on $[0,s_f]$, and the numerical solutions with time indexes near $0$ are fitted by least squares to obtain two continuous functions $\bar{H}(s)$ and $\bar{K}(s)$. Finally, these two continuous functions are shown (in Theorem \ref{VI_adp_converge}) to be close to the optimal solutions $H^*(s)$ and $K^*(s)$ in (\ref{HK_star}), under suitable conditions. This fulfills our goal.

With the above outline in mind, the rest of this section is devoted to presenting the details.

For $M\in\mathbb{Z}_+\backslash\{0\}$, define $t_j=j\Delta t$, $j=0,1,2,\cdots,M$. $\Delta t$ is the sampling interval. Assuming a measurable locally essentially bounded input $u_0(\cdot): [0,t_M]\rightarrow\mathbb{R}^{m}$ is applied to system (\ref{periodSYS}) to collect input/state data for learning, we have
\begin{equation}\label{VI_mid1}
\begin{split}
    &\frac{\mathrm{d}x^T(t)P(s)x(t)}{\mathrm{d}t} = \dot{x}^T(t)P(s)x(t)+x^T(t)P(s)\dot{x}(t)\\
    &=x^T(t)H(s,t)x(t)+2u_0^T(t)R(t)K(s,t)x(t).
\end{split}
\end{equation}
Integrating both sides of (\ref{VI_mid1}) from $t_j$ to $t_{j+1}$, by Fact \ref{Isometric}, we obtain
\begin{equation}\label{VI_mid2}
\begin{split}
    &[\tilde{x}(t_{j+1})-\tilde{x}(t_j)]^T\mathrm{vecs}(P(s)) = \\
    &\int_{t_j}^{t_{j+1}}\tilde{x}^T(t)\mathrm{vecs}(H(s,t))\mathrm{d}t+\\
    &\int_{t_j}^{t_{j+1}}\left(x^T(t)\otimes (2u_0^T(t)R(t))\right)\mathrm{vec}(K(s,t))\mathrm{d}t.
\end{split}
\end{equation}
Using (\ref{VI_Fourier_series}), we can organize (\ref{VI_mid2}) for $j=0,1,2,\cdots,M-1$ into a single linear matrix equation
\begin{equation}\label{VI_data_eqn1}
    \Theta\left[\begin{array}{c}
         \mathrm{vec}(W^{H}(s))  \\
         \mathrm{vec}(W^{K}(s))
    \end{array}\right]+E_N(s)= \Gamma_{\tilde{x}}\mathrm{vecs}(P(s)),
\end{equation}
where
\begin{equation*}
    \begin{split}
    \Theta &= \left[\Delta^T_0, \Delta^T_1, \cdots, \Delta^T_{M-1}\right]^T,\  \Delta_j = \left[I_{F\tilde{x},j}, I_{Fxu,j}\right], \\
    \Gamma_{\tilde{x}} &= \left[\delta^T_0, \delta^T_1, \cdots, \delta^T_{M-1}\right]^T, \  \delta_{j} = \tilde{x}^T(t_{j+1})-\tilde{x}^T(t_j),  \\
    E_N(s) &= \left[e_{0,N}(s),e_{1,N}(s),\cdots,e_{M-1,N}(s)\right]^T, \\
    I_{F\tilde{x},j} &=  \int_{t_j}^{t_{j+1}}F^T_N\otimes\tilde{x}^T\mathrm{d}t,\\
    I_{Fxu,j} &= \int_{t_j}^{t_{j+1}} F^T_N\otimes x^T\otimes (2u_0^TR)\mathrm{d}t, \\
    \end{split}
\end{equation*}
\begin{equation}\label{VI_Truncat_error}
    \begin{split}
    e_{j,N}(s) &= \int_{t_j}^{t_{j+1}} \tilde{x}^T(t)e_N^{H}(s,t)dt\\
    &+\int_{t_j}^{t_{j+1}}\left(x^T(t)\otimes 2u_0^T(t)\right)e_N^{K}(s,t)\mathrm{d}t.     
    \end{split}
\end{equation}
The following Assumption is imposed on the matrix $\Theta$.
\begin{assumption}\label{VI_assum}
Given $N>0$, there exist $\bar{M}>(n_1+n_2)(2N+1)$ and $\alpha>0$ (independent of $N$), such that for all $M>\bar{M}$, $M\in\mathbb{Z}_+$,
\begin{equation}\label{rank_theta}
    \frac{1}{M}\Theta^T\Theta \geq \alpha I_{(n_1+n_2)(2N+1)}.
\end{equation}
Moreover, for all $t\in[0,t_M]$, $\vert x(t)\vert\leq \beta$, $\beta$ independent of $N$.
\end{assumption}
\begin{remark}\label{remark_PE}
Condition (\ref{rank_theta}) appeared in the past literature of ADP \cite{jiang2017robust}, \cite{Frank-book1}, \cite{8100742}, which is in the spirit of persistency of excitation (PE) in adaptive control. An exploration noise, such as sum of sinusoidal signals with diverse frequencies, can be added to $u_0$, if needed, to satisfy (\ref{rank_theta}).
\end{remark}
\begin{remark}\label{remark_reset}
Notice that by derivations from (\ref{VI_mid1}) to (\ref{VI_data_eqn1}), input/state data on any time interval of length $\Delta t$ satisfying (\ref{VI_mid2}) can be used to construct (\ref{VI_data_eqn1}). To fulfill condition $\vert x(t)\vert\leq \beta$, one way is to restart the system at $\vert x(t_0)\vert\leq\beta$, whenever the boundedness condition is violated. The choice of $t_0$ depends on the problem at hand. In rotor-blade system \cite{camino2019periodic}, for example, a choice of $t_0$ can be the moment when the angular position of the rotor is zero. 
\end{remark}

The following lemma shows that equation (\ref{VI_data_eqn1}) is differentiable in $s$.
\begin{lemma}\label{VI_lemma_derivative}
$W^{H}(\cdot)$, $W^{K}(\cdot)$, ${e}^{H}_{N}(\cdot,t)$, ${e}^{K}_{N}(\cdot,t)$ and $E_N(\cdot)$ in equations (\ref{VI_Fourier_series}) and (\ref{VI_data_eqn1}) are continuously differentiable in algorithmic time $s$.
\end{lemma}
\begin{proof}
From the definition (\ref{VI_variables_def}), $H(s,t)$ and $\partial_sH(s,t)$ are continuous both in $s$ and $t$. Then by Leibniz integral rule and the definition of Fourier coefficients, we have
\begin{align}\label{VI_lemma_derivative_eqn1}
    [\dot{W}^{H}(s)]_{i,k} = \frac{2}{T}\int_{-T/2}^{T/2} [\mathrm{vecs}(\partial_sH(s,t))]_i p(k,t)dt,
\end{align}
where $i=1,2,\cdots,n_1$, $k=1,2,\cdots,2N+1$ and
$$p(k,t)=\begin{cases}
1, & \text{if } k=1 \\
\cos{(\omega tk/2)}, & \text{if } k\text{ is even} \\
\sin{(\omega t\lfloor k/2\rfloor)}, & \text{if } k\text{ is odd} \text{ and } k>1
\end{cases}.
$$
Thus by \cite[Definition 10.1]{rudin1976principles}, $W^{H}(\cdot)$ is continuously differentiable in $s$. By (\ref{VI_Fourier_series}), ${e}^{H}_{N}(\cdot,t)$ is continuously differentiable in $s$. With similar arguments, we know that $W^{K}(\cdot)$ and ${e}^{K}_{N}(\cdot,t)$ are continuously differentiable in $s$. Note that ${e}^{H}_{N}(s,t)$, ${e}^{K}_{N}(s,t)$, $\partial_s{e}^{H}_{N}(s,t)$ and $\partial_s{e}^{K}_{N}(s,t)$ are continuous both in $s$ and $t$. Again, by Leibniz integral rule, (\ref{VI_Truncat_error}) and \cite[Definition 10.1]{rudin1976principles}, $E_N(\cdot)$ is continuously differentiable in $s$. This completes the proof.
\end{proof}
By Lemma \ref{VI_lemma_derivative}, equations (\ref{PRE_HK}) and (\ref{VI_Fourier_series}), and Assumption \ref{VI_assum}, taking derivatives with respect to $s$ on both sides of (\ref{VI_data_eqn1}), we have
\begin{equation}\label{VI_data_eqn3}
    \left[\begin{array}{c}
         \mathrm{vec}(\dot{W}^{H}(s)) \\
         \mathrm{vec}(\dot{W}^{K}(s))
    \end{array}\right]=\mathcal{H}\left(W(s),s\right)+\mathcal{G}\left(W(s),s\right),
\end{equation}
where 
\begin{align*}
    &\mathcal{H}\left(W(s),s\right) =\Theta^\dagger\Gamma_{\tilde{x}}\left[ -W^{H}(s)F_N(s)-\mathrm{vecs}(C^T(s)C(s))\right.\\
    &\left.+\mathrm{vecs}\left((\mathrm{vec}^{-1}(W^{K}(s)F_N(s)))^TR(s)\right.\right.\\
    &\left.\left.\mathrm{vec}^{-1}(W^{K}(s)F_N(s))\right)
    \right], \\
    &\mathcal{G}\left(W(s),s\right) = \Theta^\dagger\left[-\dot{E}_N(s)\right.\\
    &\left.+\Gamma_{\tilde{x}}\left(
    \mathrm{vecs}\left(
    (\mathrm{vec}^{-1}(W^{K}(s)F_N(s)))^TR(s)\mathrm{vec}^{-1}(e_N^{K}(s,s))\right.\right.\right.\\
    &\left.\left.\left.+  (\mathrm{vec}^{-1}(e_N^{K}(s,s)))^TR(s)\mathrm{vec}^{-1}(W^{K}(s)F_N(s))\right.\right.\right.\\
    &\left.\left.\left.+(\mathrm{vec}^{-1}(e_N^{K}(s,s)))^TR(s)\mathrm{vec}^{-1}(e_N^{K}(s,s))
    \right)-e^{H}_N(s,s)
    \right)\right],  \\
    &W(s) = \left[\left(W^{H}(s)\right)^T, \left(W^{K}(s)\right)^T\right]^T.
\end{align*}
In (\ref{VI_data_eqn3}), all the terms containing the truncation errors are grouped into $\mathcal{G}(W(s),s)$. If $\mathcal{G}(W(s),s)$ is ignored, we obtain the following differential equation
\begin{equation}\label{VI_adp_eqn1}
     \left[\begin{array}{c}
         \mathrm{vec}(\dot{\hat{W}}^{H}(s)) \\
         \mathrm{vec}(\dot{\hat{W}}^{K}(s))
    \end{array}\right] = \mathcal{H}\left(\hat{W}(s),s\right),\qquad \hat{W}(s_f) = 0.
\end{equation}
where
$$\hat{W}(s) = \left[\left(\hat{W}^{H}(s)\right)^T, \left(\hat{W}^{K}(s)\right)^T\right]^T.$$
Notice that no explicit system dynamics information is contained in (\ref{VI_adp_eqn1}). Thus if $\hat{W}(s)$ is close to $W(s)$, approximate optimal controllers are possible to be found in view of (\ref{VI_Fourier_series}), directly from the collected data. Indeed, the following two lemmas show that $\hat{W}(s)$ can be made close to $W(s)$.
\begin{remark}\label{remark_algorithmic_time}
Through the derivations from (\ref{PRE_s}) to (\ref{VI_adp_eqn1}), equation (\ref{VI_mid1}) is a critical step in getting rid of the knowledge of system dynamics. The validity of (\ref{VI_mid1}) is a result of distinguishing the algorithmic time $s$ from the system evolution time $t$. This justifies the importance of the separation of time variables. 
\end{remark}
\begin{assumption}\label{VI_assum_differentiable}
The matrix-valued functions $A(\cdot)$ and $B(\cdot)$ are $T$-periodic and continuously differentiable on $\mathbb{R}$.
\end{assumption}
\begin{lemma}\label{VI_lemma_error_converge}
Under Assumptions \ref{structure_assum}, \ref{VI_assum} and \ref{VI_assum_differentiable}, for any $-\infty<s'<s_f$:
\begin{enumerate}
    \item $e_N^{H}(s,t)$, $e_N^{K}(s,t)$, $\partial_s e_N^{H}(s,t)$, $\partial_s e_N^{K}(s,t)$ all converge uniformly to $0$ on $[s',s_f]\times \mathbb{R}$, as $N\rightarrow \infty$.
    \item for any $ \epsilon>0$, there exists $\bar{N}>0$, such that $\forall N>\bar{N}$, 
$$\sup_{s\in[s',s_f]}\vert\mathcal{G}\left(W(s),s\right)\vert<\epsilon.$$
\end{enumerate}
\end{lemma}
\begin{proof}
See the Appendix.
\end{proof}
\begin{lemma}\label{VI_lemma_param_converge}
Let $G=0$ in PRE (\ref{PRE_s}). Under Assumptions \ref{structure_assum}, \ref{VI_assum} and \ref{VI_assum_differentiable}, for any $ \epsilon>0$ and $0<s_f<\infty$, there exists $\bar{N}>0$, such that $\forall N>\bar{N}$,
\begin{equation*}
    \begin{split}
        \sup_{s\in[0,s_f]}\Vert W^{H}(s)-\hat{W}^{H}(s)\Vert&<\epsilon,\\ \sup_{s\in[0,s_f]}\Vert W^{K}(s)-\hat{W}^{K}(s)\Vert&<\epsilon.
    \end{split}
\end{equation*}
\end{lemma}
\begin{proof}
See the Appendix.
\end{proof}
With Lemma \ref{VI_lemma_param_converge}, we can solve equation (\ref{VI_adp_eqn1}) by any convergent numerical method \cite[Section 213]{butcher2016numerical} backward in time on $[0,s_f]$, to find approximate values of $W(s)$. Supposing in the numerical method, the step size is $h$, and at step values $\{s_k\}_{k=0}^{L}$, $s_k = kh$, $hL=s_f$, the numerical solutions of (\ref{VI_adp_eqn1}) are computed, denoted by $\{\hat{W}^{H}_k\}_{k=0}^L$ and $\{\hat{W}^{K}_k\}_{k=0}^L$. Then we can define
\begin{equation}\label{HK_hat}
        \mathrm{vecs}(\hat{H}_k) = \hat{W}^{H}_kF_N(s_k),\quad \mathrm{vec}(\hat{K}_k) = \hat{W}^{K}_kF_N(s_k),
\end{equation}
which are used as approximations to $H(s,s)$ and $K(s,s)$ at $\{s_k\}_{k=0}^{L}$, respectively. But (\ref{HK_hat}) is discrete, which is not convenient. So we would like to fit the part of the data close to (\ref{HK_star}) in (\ref{HK_hat}) to get two continuous functions $\bar{H}(\cdot)$ and $\bar{K}(\cdot)$. Supposing we are able to choose a $\bar{L}\in\mathbb{Z}_+$, satisfying $s_{\bar{L}}>T$ and $\lfloor L/2\rfloor>\bar{L}>2N+1$, define
\begin{align*}
    &\mathcal{U}= \left[\begin{array}{cccc}
         F_N(s_0) & F_N(s_1) & \cdots & F_N(s_{\bar{L}})
    \end{array}\right]^T,\\
    &\mathcal{V}=\left[\begin{array}{cccc}
         \mathrm{vecs}(\hat{H}_0)  &
         \mathrm{vecs}(\hat{H}_1)  &
         \cdots &
         \mathrm{vecs}(\hat{H}_{\bar{L}})
    \end{array}\right]^T, \\
    &\mathcal{W}=\left[\begin{array}{cccc}
         \mathrm{vec}(\hat{K}_0)  &
         \mathrm{vec}(\hat{K}_1)  &
         \cdots &
         \mathrm{vec}(\hat{K}_{\bar{L}})
    \end{array}\right]^T. \\
\end{align*} 
\begin{assumption}\label{rank_U}
Given $N>0$, there exist $\lfloor L/2\rfloor>\bar{L}_0>2N+1$ and $\alpha>0$ (independent of $N$), such that for all $\lfloor L/2\rfloor>\bar{L}>\bar{L}_0$, $s_{\bar{L}}>T$,
$$\frac{1}{\bar{L}}\mathcal{U}^T\mathcal{U}\geq \alpha I_{2N+1}.$$
\end{assumption}
Under Assumption \ref{rank_U}, over-determined least squares fittings are implemented on data sets $\{\mathcal{V},\mathcal{U}\}$ and $\{\mathcal{W},\mathcal{U}\}$, respectively, to get
\begin{equation}\label{HK_bar}
\begin{split}
    \bar{H}(t) &= \mathrm{vecs}^{-1}(\bar{W}^{H}F_N(t)),\\ \bar{K}(t) &= \mathrm{vec}^{-1}(\bar{W}^{K}F_N(t)),
\end{split}
\end{equation}
where
\begin{equation}\label{VI_final_fit}
    (\bar{W}^{H})^T=\mathcal{U}^\dagger\mathcal{V},\qquad (\bar{W}^{K})^T=\mathcal{U}^\dagger\mathcal{W}.
\end{equation}
We are in a position to state our main result on approximating the optimal solution of the infinite-horizon PLQ problem without the precise knowledge of system dynamics.
\begin{theorem}\label{VI_adp_converge}
Consider the infinite-horizon PLQ optimal control problem of system (\ref{periodSYS}) with cost ($\ref{cost_inf}$). Under Assumptions \ref{structure_assum}, \ref{VI_assum}, \ref{VI_assum_differentiable} and \ref{rank_U}, for any $\epsilon>0$, there exist $\bar{s}_f>0$, $\bar{N}>0$, $\bar{h}>0$, such that $\forall s_f>\bar{s}_f$, $\forall N>\bar{N}$, any $0<h<\bar{h}$, we have
$$\sup_{t\in\mathbb{R}}\Vert\bar{H}(t)-H^*(t)\Vert<\epsilon,\quad \sup_{t\in\mathbb{R}}\Vert\bar{K}(t)-K^*(t)\Vert<\epsilon,$$
where $\bar{L}$ is chosen to satisfy $s_{\bar{L}}>T$, $\lfloor L/2\rfloor>\bar{L}>2N+1$.
\end{theorem}
\begin{proof}
See the Appendix.
\end{proof}
To sum up, our novel VI-based off-policy ADP algorithm is presented in Algorithm \ref{VI_Alg_off}.

\begin{algorithm2e}
 \begin{minipage}{.9\linewidth}
Choose $\Delta t>0$, large enough $M>0$, $N>0$, $s_f>0$, and small enough $h>0$.\newline
Apply $u_0(\cdot): [0,t_M]\rightarrow\mathbb{R}^{m}$ (with exploration noise) to system (\ref{periodSYS}), collect the input/state data.\newline
Construct the data matrices $\Theta$ and $\Gamma_{\tilde{x}}$.\newline
Solve numerically (\ref{VI_adp_eqn1}) backward in time on $[0,s_f]$. \newline
Choose $\bar{L}$ satisfying $s_{\bar{L}}>T$, $\lfloor L/2\rfloor>\bar{L}>2N+1$.\newline
Solve (\ref{HK_bar}) for $\bar{K}(t).$ \newline
Use $\bar{u}(t) = -\bar{K}(t)x(t)$ as the approximate optimal control for all $t\in[0,\infty)$.
\caption{VI-based off-policy ADP} \label{VI_Alg_off}
 \end{minipage}
\end{algorithm2e}
\begin{figure}[!ht]
\centerline{\includegraphics[width=\columnwidth]{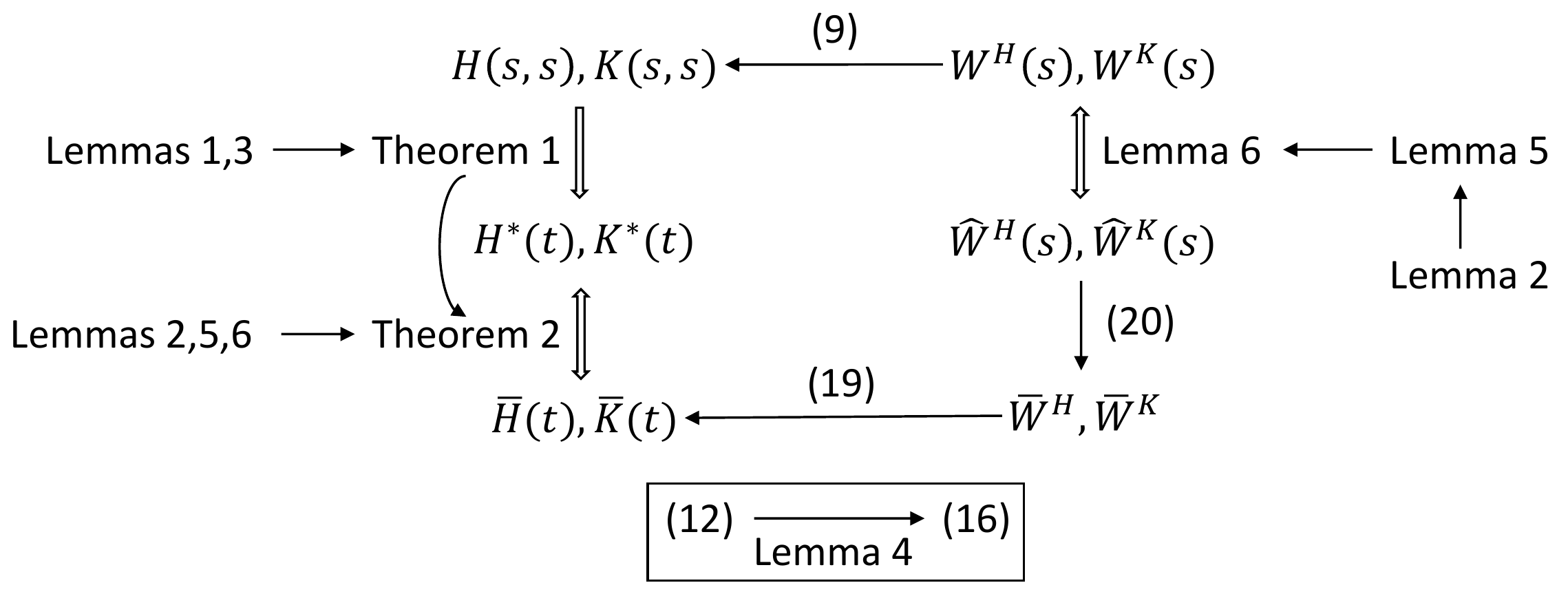}}
\caption{Overview of derivations and convergence analysis of Algorithm \ref{VI_Alg_off}.}
\label{Overview}
\end{figure}
To give a clearer explanation of the derivations and convergence analysis of Algorithm \ref{VI_Alg_off}, the relationships of different components in this section are summarized in Figure \ref{Overview}.
\begin{remark}\label{Choose_parameter}
In (\ref{VI_variables_def}) and (\ref{VI_Fourier_series}), if $\lim_{s\rightarrow-\infty}P(s) - P^*(s) = 0$, $W(s)$ will converge to a unique periodic orbit. In view of Lemma \ref{VI_lemma_param_converge}, $\hat{W}(s)$ will be close to that periodic orbit. This suggests a heuristic procedure for the choice of parameters in Algorithm \ref{VI_Alg_off}:
\begin{enumerate}[label=(\arabic*),ref=(\arabic*)]
    \item Choose some $\Delta t>0$, $N>0$.
    \item\label{Choice_M} In view of Assumption \ref{VI_assum}, $\Theta$ should be at least full column rank. Thus $M\gg (n_1+n_2)(2N+1)$. Then $t_M = M\Delta t$.
    \item Find an $s_f$, such that $\Vert\hat{W}(s)\Vert$ given by (\ref{VI_adp_eqn1}) is almost periodic near $s=0$. If such an $s_f$ can not be found for large values of $s_f$, increase $N$ and go back to step \ref{Choice_M}.
    \item Set $h>0$ such that $\lfloor \frac{s_f}{h}\rfloor\gg 2N+1$.
    \item Choose $\bar{L}$ satisfying $s_{\bar{L}}>T$, $\lfloor L/2\rfloor>\bar{L}>2N+1$.
\end{enumerate}
\end{remark}

\section{Simulation Results}\label{section_Sim}
Consider the triple inverted pendulum with a periodically varying load \cite{pendulum}, modeled by system (\ref{periodSYS})  with
\begin{align}
    &A(t) = \left[\begin{array}{cc}
        0_3 & I_3 \\
        A_{21}(t) & A_{22}
    \end{array}\right],
    B(t) = \left[\begin{array}{c}
         0_3  \\
         B_2 
    \end{array}\right]\label{triple} \\
    &A_{21}(t) = \left[\begin{array}{ccc}
        (\gamma-3) & (3-\gamma) & -1 \\
        (4-\gamma) & 2(\gamma-3) & (3-\gamma)\\
        -1 & (4-\gamma) & (\gamma-3) 
    \end{array}\right], \nonumber
\end{align}
\begin{align*}
    &A_{22} = 0.5\left[\begin{array}{ccc}
        -1 & 0 & 0 \\
        1 & -1 & 0 \\
        0 & 1 & -1
    \end{array}\right],
    B_2 = \left[\begin{array}{ccc}
        1 & -1 & 0 \\
        -1 & 2 & -1 \\
        0 & -1 & 2
    \end{array}\right],\nonumber
\end{align*}
and states $x(t)=\left[\eta_1(t),\eta_2(t),\eta_3(t),\dot{\eta}_1(t),\dot{\eta}_2(t),\dot{\eta}_3(t)\right]^T$, inputs $u(t)=[u_1(t),u_2(t),u_3(t)]^T$, where $\gamma=1+2cos(t)$. For each $i=1,2,3$, $\eta_i(\cdot)$ is the angle of the $i$th pendulum with respect to the vertical line; $\dot{\eta}_i(\cdot)$ is the corresponding angular velocity; $u_i(\cdot)$ is the control torque applied at the bottom of the $i$th pendulum.

If the system matrices in (\ref{triple}) are known, (\ref{PRE_s}) can be solved to obtain a suboptimal controller, which is referred as model-based PLQ (MBPLQ) controller. However, the actual system need not evolve exactly as (\ref{triple}). Suppose an extra periodically varying disturbance exists, which changes $A_{21}(t)$ in (\ref{triple}) to 
$$\tilde{A}_{21}(t)=A_{21}(t) + \zeta(1+sin(3t))I_3,$$
where $\zeta>0$ controls the magnitude of the disturbance. Algorithm 1 is applied to (\ref{triple}) with $\tilde{A}_{21}(t)$, to obtain ADP controllers. The simulation results with different choices of controllers, $N$, $M$, $s_f$ and $\zeta$ are summarized in Table \ref{table_simulation}. It is easy to check that Assumptions \ref{structure_assum} and \ref{VI_assum_differentiable} are satisfied by  system (\ref{triple}) and the choice of $C(\cdot)$ in Table \ref{table_simulation}. In the simulation trials of ADP controllers, the following initial controller
$$[u_0(t)]_i = 0.2*\sum_{j=1}^{500} \sin{(\omega_{i,j}t)},\quad i=1,2,3$$
is applied to the system over time interval $[0,t_M]$ to collect data, where $\omega_{i,j}$ is drawn from a uniform distribution over $[-500,500]$. To guarantee the collected data is bounded, we reset $x(t_j)$ to the initial state $x(t_0)=0$, $t_0=0$ as long as $\vert x(t_j)\vert>10$. In Table \ref{table_simulation}, trial 1 yields the best result. Loss of stability and/or optimality occurs when one of the parameters $N$, $M$, $s_f$ is not large enough, as shown in trials 2-6. Although MBPLQ controller is robust to extra disturbance with small magnitude in Trial 7, it is not robust to extra disturbance of large magnitude in trial 8, in contrast to the ADP controller. The different control gains generated by Algorithm \ref{VI_Alg_off} in trial 1 are compared with the optimal gains in Fig. \ref{Gains}. These simulation results demonstrate the viability of the theoretical results in the previous section.
\begin{table*}[h!]
\centering
\caption{simulation results of the triple inverted pendulum control under different settings}
\label{table_simulation}
 \begin{tabular}{| c | c | c | c | c | c | c | c | c |} 
 \hline
 Trials & Controller & $N$ & $M$ & $s_f$ & $\zeta$ & Number of resetting & Stability & $\max_t\Vert \bar{K}(t)-K^*(t)\Vert$   \\
 \hline
 1 & ADP & 6 & 800 & 40 & 1 & 24 & Yes & 0.0498 \\ 
 \cline{1-9}
 2 & ADP & 3 & 800 & 40 & 1 & 24 & Yes & 0.8784 \\
 \cline{1-9}
 3 & ADP & 1 & 800 & 40 & 1 & 25 & No & 64.9159 \\
 \cline{1-9}
 4 & ADP & 6 & 800 & 12 & 1 & 25 & Yes & 8.1861 \\
 \cline{1-9}
 5 & ADP & 6 & 800 & 8 & 1 & 24 & No & 146.4805\\
 \cline{1-9}
 6 & ADP & 6 & 400 & 40 & 1 & 13 & No & 151.3061 \\
 \cline{1-9}
 7 & MBPLQ & - & - & - & 0.1 & - & Yes & - \\
 \cline{1-9}
 8 & MBPLQ & - & - & - & 1 & - & No & - \\
 \hline
 Common parameters & \multicolumn{8}{c|}{$\Delta t=0.2$, $h=0.1$, $C = I_6$, $R=I_3$, $\bar{L}=\lfloor s_f/(3h) \rfloor$.}\\
 \hline
\end{tabular}
\end{table*}
\begin{figure}[!ht]
\centerline{\includegraphics[width=0.9\columnwidth]{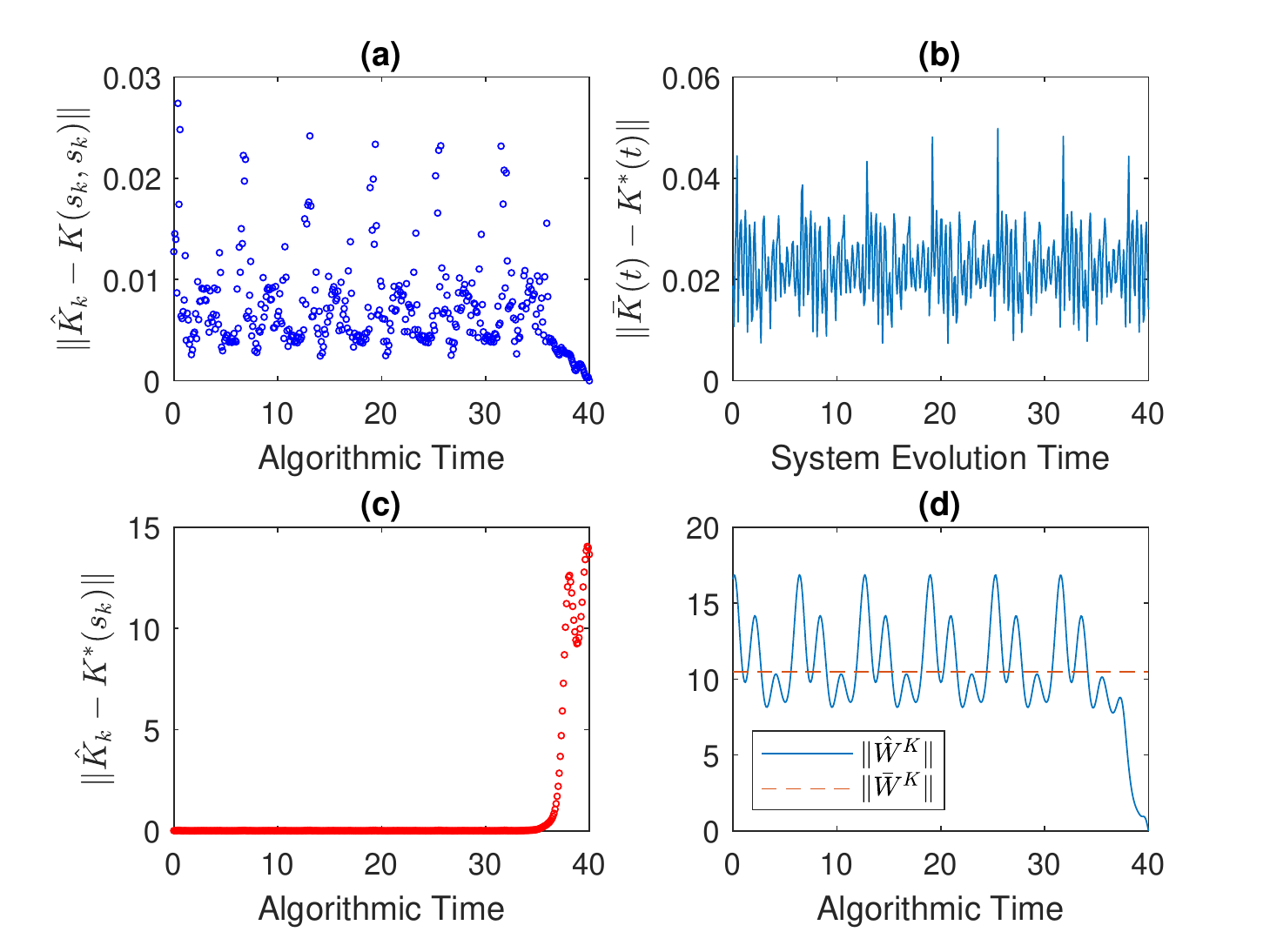}}
\caption{Comparison of different control gains. $\bar{K}(\cdot)$ is the output of Algorithm \ref{VI_Alg_off}; $\hat{K}_k$ is defined in (\ref{HK_hat}); $K(\cdot)$ is generated by model-based VI (\ref{PRE_s}); $K^*(\cdot)$ is the optimal control gain; $\hat{W}^{K}$ is generated by (\ref{VI_adp_eqn1}); $\bar{W}^{K}$ is given in (\ref{VI_final_fit}).
}
\label{Gains}
\end{figure}

\section{Conclusion}\label{section_conclusion}
An innovative VI-based ADP algorithm is proposed for CTLP systems in this paper, such that learning-based suboptimal controllers can be obtained from real-time data without the exact knowledge of system dynamics. The proposed algorithm does not assume an initial stabilizing controller, and is off-policy, which is easy-to-use and data-efficient. Convergence analysis is developed for the presented VI-based ADP algorithm. It is shown that, under mild conditions,  the proposed algorithm generates a sequence of suboptimal controllers converging uniformly to the optimal solutions. In addition, the proposed adaptive optimal control method is successfully tested in a benchmark example of controlling a triple inverted pendulum with a periodically varying load. Our future work can be directed at extending the proposed methodology to the optimal output regulation problem as shown in \cite{8100742}.

\section*{Acknowledgment}
It is a pleasure to thank the Associate Editor, anonymous reviewers and Tao Bian for their helpful and constructive comments.

\appendix
\subsection{Proof of Theorem \ref{VI}}
The proof is divided into three cases: Case 1) $G >0$; Case 2) $G=0$; Case 3) $G\geq 0$, $G\neq 0$ with at least one zero eigenvalue.

Case 1): In this case, by Lemma \ref{UniqueExist} and \cite[Corollary]{doi:10.1080/00207179208934305}, we immediately obtain
$\lim_{t\rightarrow-\infty}\left(P(t;t_f,G)-P^*(t)\right)=\lim_{t\rightarrow-\infty}\left(P(t;t_f,G)-P_S(t)\right)=0.$

Case 2): In this case, for any $t\leq \tau_1\leq \tau_2$ and any $ \xi\in \mathbb{R}^n$, define $u^{(2)}(s)\triangleq-R(s)^{-1}B(s)^TP(s;\tau_2,0)x^{(2)}(s)$, where $x^{(2)}(s), t \leq s\leq \tau_2$, is the solution of system (\ref{periodSYS}) under control $u^{(2)}(s)$. By definition,
\begin{equation*}
    \begin{split}
        \xi^TP(t;\tau_1,0)\xi&=\min_u V(t,\tau_1,\xi,u,0)\leq  V(t,\tau_1,\xi,u^{(2)},0)\\&
        \leq V(t,\tau_2,\xi,u^{(2)},0)= \xi^TP(t;\tau_2,0)\xi.
    \end{split}
\end{equation*}
Hence $P(t;\tau,0)$ is nondecreasing as $\tau\rightarrow\infty$, when $t\leq\tau$. By Lemma \ref{UniqueExist}, $P^*(t)$ is stabilizing solution of the PRE, which means $\xi^TP^*(t)\xi<\infty$, for any fixed $\xi\in \mathbb{R}^n$ and $t$. Again, by definition,
\begin{align*}
    \xi^TP(t;\tau,0)\xi&\leq V(t,\tau,\xi,u^*,0)\\
    &\leq 
    \lim\limits_{\bar{\tau}\rightarrow\infty}V(t,\bar{\tau},\xi,u^*,0)=\xi^TP^*(t)\xi<\infty.
\end{align*}
Therefore when $t\leq\tau$, $P(t;\tau,0)$ is nondecreasing as $\tau\rightarrow\infty$, and bounded from the above. By the monotone convergence theorem, $\bar{P}(t;0):=\lim_{\tau\rightarrow\infty}P(t;\tau,0)$ exists. By periodicity,
$$P(t+T;\tau+T,0)=P(t;\tau,0).$$
Then we have
$$\bar{P}(t+T;0)=\lim_{\tau\rightarrow\infty}P(t+T;\tau,0)=\lim_{\tau\rightarrow\infty}P(t;\tau,0)=\bar{P}(t;0).$$
Furthermore, $P(t;\tau,0)$ is symmetric and positive semidefinite for all $t\leq\tau$. This implies that $\bar{P}(t;0)$ is a SPPS solution of the PRE. Due to the uniqueness of the SPPS solution, we have $\bar{P}(t;0)=P^*(t)$. In $P(t;\tau,0)$, $\tau\rightarrow\infty$ is equivalent to $t\rightarrow-\infty$, thus (\ref{VImodel}) holds when $G=0$.

Case 3): In this case,we can always find a $\bar{G}>0$, such that $0\leq G\leq\bar{G}$, then (\ref{VImodel}) follows by Lemma \ref{VI_monotone} and the squeeze theorem \cite[Theorem 3.3.6]{sohrab2003basic}.
\subsection{Proof of Lemma \ref{VI_lemma_error_converge}}
For 1), by Assumption \ref{VI_assum_differentiable}, $A(t)$ is locally Lipschitz continuous at $t$, hence $A(\cdot)$ is Lipschitz continuous on compact set $[t,t+T]$. Due to the periodicity, we know that $A(\cdot)$ is Lipschitz continuous on $\mathbb{R}$. Define
$$\mathcal{F}(s,\tau)=\frac{\mathrm{vecs}(H(s,t-\tau))-\mathrm{vecs}(H(s,t))}{\sin(\omega\tau/2)}$$
for $0<\vert\tau\vert\leq \frac{T}{2}$, and put $\mathcal{F}(s,0)=0$. Then we have
$$\vert\mathcal{F}(s,\tau)\vert\leq\frac{2\Vert A(t-\tau)-A(t)\Vert\Vert P(s)\Vert}{\vert\sin(\omega\tau/2)\vert}\leq \bar{U}\left\vert\frac{\tau}{\sin(\omega\tau/2)}\right\vert,$$
where $\bar{U}>0$ does not depend on $t$, $s$ or $\tau$, since $A(\cdot)$ is Lipschitz continuous on $\mathbb{R}$ and $P(s)$ is bounded by Theorem \ref{VI}. From above inequalities, it is easy to see that $\vert\mathcal{F}(s,\tau)\vert$ is bounded on $[s',s_f]\times [-T/2,T/2]$. Following the same derivations as those in \cite[Theorem 8.14]{rudin1976principles}, we have
\begin{equation*}
    \begin{split}
        &W^{H}(s)F_N(t)-\mathrm{vecs}(H(s,t))=\\
        &\frac{1}{T}\int_{-T/2}^{T/2}\left[\mathcal{F}(s,\tau)\cos(\frac{\omega\tau}{2})\right]\sin(N\omega\tau)\mathrm{d}\tau\\
        &+\frac{1}{T}\int_{-T/2}^{T/2}\left[\mathcal{F}(s,\tau)\sin(\frac{\omega\tau}{2})\right]\cos(N\omega\tau)\mathrm{d}\tau,
    \end{split}
\end{equation*}
which converges uniformly on $[s',s_f]\times \mathbb{R}$ to $0$ as $N\rightarrow\infty$, as a result of \cite[Theorem 8.12]{rudin1976principles} and the boundedness of $\mathcal{F}(s,\tau)$. Therefore, $\lim_{N\rightarrow\infty}e_N^{H}(s,t) = 0$ uniformly on $[s',s_f]\times \mathbb{R}$. Through similar arguments, we can prove that $e_N^{K}(s,t)$, $\partial_s e_N^{H}(s,t)$, $\partial_s e_N^{K}(s,t)$ all converge uniformly to $0$ on $[s',s_f]\times \mathbb{R}$.

For 2), by definition of $\mathcal{G}(W(s),s)$, Lemma \ref{FourierConverge} and Assumption \ref{VI_assum}, for any $\epsilon_0>0$, there exists $\bar{N}_0>0$, such that $\forall N>\bar{N}_0$,
\begin{equation*}
    \begin{split}
        \sup_{s\in[s',s_f]}\vert\mathcal{G}\left(W(s),s\right)\vert&<M\Vert\Theta^\dagger\Vert_2(1+2\beta^4)\epsilon_0\\&=\frac{M(1+2\beta^4)}{\sigma_{min}(\Theta)}\epsilon_0\leq\frac{1+2\beta^4}{\alpha}\epsilon_0.
    \end{split}
\end{equation*}
Note that $\alpha$ and $\beta$ are independent of $N$, the proof is completed.

\subsection{Proof of Lemma \ref{VI_lemma_param_converge}}
Firstly, we prove that if the solution of (\ref{VI_adp_eqn1}) exists on interval $[s',s_f]$, for some $-\infty<s'<s_f$ and $\forall N>N_1$, $N_1>0$, then Lemma $\ref{VI_lemma_param_converge}$ holds on this interval. Subtracting (\ref{VI_adp_eqn1}) from (\ref{VI_data_eqn3}), we have

\begin{equation}\label{VI_lemma_param_converge_ode1}
\begin{split}
     \left[\begin{array}{c}
         \mathrm{vec}(\dot{Z}^{(1)}(s)) \\
         \mathrm{vec}(\dot{Z}^{(2)}(s))
    \end{array}\right] &= \mathcal{H}\left(W(s),s\right) \\
    &-\mathcal{H}\left(\hat{W}(s),s\right)+\mathcal{G}\left(W(s),s\right),    
\end{split}
\end{equation}
where
$$Z(s) = \left[\begin{array}{c}
     Z^{H}(s)  \\
     Z^{K}(s) 
\end{array}\right]=\left[\begin{array}{c}
     W^{H}(s)-\hat{W}^{H}(s)  \\
     W^{K}(s)-\hat{W}^{K}(s) 
\end{array}\right],$$
$Z(s_f)=0$. Consider the following differential equation evolving on $[s',s_f]$
\begin{equation}\label{VI_lemma_params_converge_ode2}
    \left[\begin{array}{c}
    \mathrm{vec}(\dot{Z}_0^{H}(s)) \\
    \mathrm{vec}(\dot{Z}_0^{K}(s))
    \end{array}\right] = \mathcal{H}\left(W_0(s),s\right)-\mathcal{H}\left(\hat{W_0}(s),s\right),
\end{equation}
with $Z_0(s_f)=0$. Obviously, it admits a solution $Z_0(\cdot)\equiv 0$.

On one hand, we know from Lemma \ref{VI_lemma_error_converge} that for any $\epsilon>0$,  there is some $N_\epsilon$, such that $\sup_{s\in[s',s_f]}\vert\mathcal{G}\left(W(s),s\right)\vert<\epsilon$. On the other hand, note that the RHS of (\ref{VI_lemma_param_converge_ode1}) and RHS of (\ref{VI_lemma_params_converge_ode2}) are locally Lipschitz in $Z(s)$ and $Z_0(s)$, respectively. Then by \cite[Theorem 55]{sontag2013mathematical}, we have 
\begin{equation}\label{Difference_inequal}
    \sup_{s\in[s',s_f]}{\Vert Z(s)\Vert}<g(\epsilon),
\end{equation}
where $g(\cdot):\mathbb{R}^+\rightarrow\mathbb{R}^+$ is a class $\mathcal{K}_\infty$ function \cite[Definition 4.2]{khalil1996noninear}. Thus, by choosing $\epsilon$ arbitrarily small, $\sup_{s\in[s',s_f]}{\Vert Z(s)\Vert}$ can also be made arbitrarily small.

Now, we prove that the solution of (\ref{VI_adp_eqn1}) indeed exists on $[0,s_f]$ for $N$ large enough. This amounts to proving that $Z(s)$ exists on $[0,s_f]$ for $N$ large enough, because $W^{H}(s)$ and $W^{K}(s)$ always exist on $(-\infty,s_f]$. For a fixed $N\in\mathbb{Z}_+$, since the RHS of (\ref{VI_lemma_param_converge_ode1}) is continuous in $s$ and locally Lipschitz at $Z(s_f)=0$, by \cite[Theorem 3.1]{khalil1996noninear}, (\ref{VI_lemma_param_converge_ode1}) has a unique solution on $[s_f-\delta, s_f]$, for some $\delta>0$. Let $(S_N, s_f]$ be the maximal interval of existence of (\ref{VI_lemma_param_converge_ode1}). If $S_N=-\infty$, then we are done. Otherwise, $\lim\limits_{s\rightarrow S_N^+}\Vert Z(s)\Vert=\infty$. In the latter case, we claim that $\lim\limits_{N\rightarrow+\infty}S_N=-\infty$. To see this, let $N_1,N_2\in\mathbb{Z}_+$ and $N_1<N_2$. Supposing $S_{N_1}<S_{N_2}$, by (\ref{Difference_inequal}) and Lemma \ref{VI_lemma_error_converge}, for (\ref{VI_lemma_param_converge_ode1}) with $N=N_2$, $\lim\limits_{s'\rightarrow S_{N_2}^+}\sup_{s\in[s',s_f]}{\Vert Z(s)\Vert}<c$, where $c>0$ is some finite constant. This contradicts to the fact that $S_{N_2}$ is a finite escape time for (\ref{VI_lemma_param_converge_ode1}) with $N=N_2$. Thus we can only have $S_{N_1}\geq S_{N_2}$, i.e., $\{S_N\}_{N=0}^\infty$ is a non-increasing sequence. Assuming $\lim\limits_{N\rightarrow\infty}S_N=S>-\infty$, we have $\lim\limits_{s\rightarrow S^+}\Vert Z(s)\Vert = \infty$. But by (\ref{Difference_inequal}) and Lemma \ref{VI_lemma_error_converge}, for any $S<s'\leq s_f$,
$$\lim\limits_{N\rightarrow\infty}\left(\sup_{s\in[s',s_f]}{\Vert Z(s)\Vert}\right)=0.$$
We arrive at another contradiction. Thus the only possibility is $\lim\limits_{N\rightarrow+\infty}S_N=-\infty$, which means the solution of (\ref{VI_adp_eqn1}) exists on $[0,s_f]$ for $N$ large enough. The proof is thus completed.

\subsection{Proof of Theorem \ref{VI_adp_converge}}
For convenience, only the convergence of $\bar{H}$ to $H^*$ is proved. The proof on the convergence of $\bar{K}$ to $K^*$ is analogous. Set $G=0$. For any $\epsilon_0>0$, by Lemma \ref{VI_lemma_error_converge}, there exists $\bar{N}_1>0$, such that $\forall N>\bar{N}_1$,
\begin{equation*}
        \sup_{s\in[0,s_f]}\Vert \mathrm{vecs}^{-1}(W^{H}(s)F_{N}(s))-H(s,s)\Vert<\frac{\epsilon_0}{4}.
\end{equation*}
By Lemma \ref{VI_lemma_param_converge}, there exists $\bar{N}_2>0$, such that $\forall N>\bar{N}_2$,
\begin{equation*}
        \sup_{s\in[0,s_f]}\vert \hat{W}^{H}(s)F_N(s)-W^{H}(s)F_N(s)\vert<\frac{\epsilon_0}{4}.
\end{equation*}
According to \cite[Theorem 213B]{butcher2016numerical}, we can choose a $\bar{h}_0$, such that for all step size $0<h<\bar{h}_0$,
\begin{equation*}
        \sup_{k\in\{1,2,\cdots,\bar{L}\}}\vert \mathrm{vecs}(\hat{H}_k)-\hat{W}^{H}(s_k)F_N(s_k)\vert<\frac{\epsilon_0}{4}.
\end{equation*}
From Theorem \ref{VI}, there exists a $\bar{s}_{f,0}>0$, such that $\forall s_f>\bar{s}_{f,0}$,
\begin{equation*}
        \sup_{s\in[0,s_{\bar{L}}]}\Vert H(s,s)-H^*(s)\Vert<\frac{\epsilon_0}{4}.
\end{equation*}
Set $\bar{N}_0 = \max\{\bar{N}_1,\bar{N}_2\}$. Applying the triangle inequality to the above inequalities and by Fact \ref{Isometric}, we obtain that, $\forall s_f>\bar{s}_{f,0}$, $N>\bar{N}_0$, any $0<h<\bar{h}_0$, there are
\begin{equation}\label{VI_adp_converge_eqn1}
    \sup_{k\in\{1,2,\cdots,\bar{L}\}}\Vert \hat{H}_k-H^*(s_k)\Vert<\epsilon_0.   
\end{equation}
Now define
\begin{equation*}
    \mathcal{V}^*=\left[
         \mathrm{vecs}(H^*(s_0)),
         \mathrm{vecs}(H^*(s_1)),
         \cdots ,
         \mathrm{vecs}(H^*(s_{\bar{L}}))
    \right]^T,
\end{equation*}
and express $H^*(s)$ by their Fourier series,
\begin{equation*}
    \mathrm{vecs}(H^*(s)) = W^*_1F_N(s)+e^*_{1,N}(s).
\end{equation*}
In view of (\ref{VI_final_fit}) and Assumption 4, we have
\begin{equation*}
    (W^*_1)^T=\mathcal{U}^\dagger(\mathcal{V^*}+E^*_{1,N}).
\end{equation*}
Thus we obtain
\begin{equation*}
    (W^*_1-\bar{W}^{H})^T=\mathcal{U}^\dagger\left[(\mathcal{V}^*-\mathcal{V})+E^*_{1,N}\right].
\end{equation*}
By the properties of matrix norms, we have
\begin{equation*}
    \begin{split}
    \frac{1}{\sqrt{n_1}}\Vert(W^*_1-\bar{W}^{H})^T\Vert&\leq \Vert(W^*_1-\bar{W}^{H})^T\Vert_2  \\
    &\leq \Vert\mathcal{U}^\dagger\Vert_2\Vert\left(\mathcal{V}^*-\mathcal{V})+E^*_{1,N}\right\Vert.
    \end{split}
\end{equation*}
By (\ref{VI_adp_converge_eqn1}), Lemma \ref{FourierConverge} and Assumption \ref{rank_U}, for any $\epsilon_3>0$, there exist large enough $s_{f,3}$, $N_3$, and small enough $h_3$, such that
\begin{equation*}
    \begin{split}
    \Vert W^*_1-\bar{W}^{H}\Vert&< \frac{\sqrt{n_1}\bar{L}}{\sigma_{\min}(\mathcal{U})}\epsilon_3< \frac{\sqrt{n_1}}{\alpha}\epsilon_3.
    \end{split}
\end{equation*}
As a result of the boundedness of $F_N(\cdot)$, for any $\epsilon>0$, there exist $\bar{s}_f>0$, $\bar{N}>0$, $\bar{h}>0$, such that $\forall s_f>\bar{s}_f$, $\forall N>\bar{N}$, any $0<h<\bar{h}$,
\begin{align*}
    &\sup_{t\in\mathbb{R}}\Vert\bar{H}(t)-\mathrm{vecs}^{-1}(W^*_1F_N(t))\Vert<\frac{\epsilon}{2},\\
    &\sup_{t\in\mathbb{R}}\Vert \mathrm{vecs}^{-1}(W^*_1F_N(t))-H^*(t)\Vert<\frac{\epsilon}{2}.
\end{align*}
Again, using the triangle inequality completes the proof.


\bibliographystyle{IEEEtran} 
\bibliography{bpangRef} 

\end{document}